\documentclass[a4paper,UKenglish,cleveref, autoref, thm-restate]{lipics-v2019}
\usepackage{tikz-cd}

\title{Quadratic type checking for objective type theory}

\author{Benno van den Berg}{Institute for Logic, Language and Computation (ILLC), University of Amsterdam, P.O. Box 94242, 1090 GE Amsterdam, The Netherlands}{b.vandenberg3@uva.nl}{https://orcid.org/0000-0002-0469-0788}{}

\author{Martijn den Besten}{Institute for Logic, Language and Computation (ILLC), University of Amsterdam, P.O. Box 94242, 1090 GE Amsterdam, The Netherlands}{m.a.denbesten@uva.nl}{}{}

\authorrunning{B. van den Berg and M. den Besten}

\Copyright{Benno van den Berg and Martijn den Besten}

\ccsdesc[500]{Theory of computation~Type theory}

\keywords{Homotopy type theory, Polynomial time algorithms, categorical semantics}

\EventEditors{John Q. Open and Joan R. Access}
\EventNoEds{2}
\EventLongTitle{42nd Conference on Very Important Topics (CVIT 2016)}
\EventShortTitle{CVIT 2016}
\EventAcronym{CVIT}
\EventYear{2016}
\EventDate{December 24--27, 2016}
\EventLocation{Little Whinging, United Kingdom}
\EventLogo{}
\SeriesVolume{42}
\ArticleNo{23}

\nolinenumbers
\hideLIPIcs

\begin{document}

\maketitle

\begin{abstract}
We introduce a modification of standard Martin-L\"of type theory in which we eliminate definitional equality and replace all computation rules by propositional equalities. We show that type checking for such a system can be done in quadratic time and that it has a natural homotopy-theoretic semantics.
\end{abstract}

\section{Introduction}

Among all formal systems for constructive mathematics, Martin-L\"of's constructive type theory has a special status in that it successfully embodies many important constructivist ideas. Among them is the idea that the meaning of a mathematical statement is fully explained by what counts as a proof of (or evidence for, or a construction of) that statement. In fact, for all practical purposes a mathematical statement can be identified with the type (or set) of its proofs: this is what is meant by the \emph{propositions as types} idea. As a result, in type theory
\[ a \in \sigma \]
can both be read as saying that $a$ is an element of the set $\sigma$ and that $a$ is a proof of the proposition $\sigma$.

The next question is whether the relation $a \in \sigma$ should be decidable (possibly appearing in a context $\Gamma$, but we will ignore that for now). This is the case in Martin-L\"of's type theory and, indeed, from a philosophical point of view, this seems desirable. If $a$ is a proof of (or evidence for) $\sigma$, then seeing $a$ should convince any rational being of the truth of $\sigma$, and we would expect any rational being to be sufficiently reflective to recognise that $a$ makes assent to $\sigma$ inevitable. Put differently, if it is reasonable to simultaneously be presented with $a$ and be unsure about the truth of $\sigma$, then $a$ is simply not conclusive evidence in favour of $\sigma$. As summarised by Kreisel: ``we can recognize a proof when we see one''.\footnote{Full quote: ``The \emph{sense} of a mathematical assertion denoted by a linguistic object $A$ is intuitionistically determined (or understood) if we have laid down what constructions constitute a \emph{proof} of $A$, i.e., if we have a construction $r_A$ such that, for any construction $c$, $r_A(c) = 0$ if $c$ is a proof of $A$ and $r_A(c) = 1$ if $c$ is not a proof of $A$: the logical particles in this explanation are interpreted truth functionally, since we are adopting the basic intuitionistic idealization that we can recognize a proof when we see one, and so $r_A$ is decidable. (Note that this applies to \emph{proof}, not \emph{provability}.)'' \cite{kreisel62}}

But also from the point of view of implementing type theory it is desirable to have decidable type checking. Indeed, most proof assistants based on type theory use type checking as their mechanism for verifying proofs. In this connection it should be stressed that we are concerned here with derivability, not with validity of derivations. Whereas for many systems the question whether some syntactic object is a valid derivation in the system or not is decidable, type theory has in addition the special property that the question whether $a \in \sigma$ is \emph{derivable} or not is decidable.

But once one has become convinced that we should make this question decidable, a natural next question is how difficult from a complexity-theoretic point of view it should be to check whether $a \in \sigma$ holds. It turns out that in most versions of type theory, including most standard formulations of Martin-L\"of type theory, the worst-case upper bounds are quite horrendous. In fact, we have the following theorem by Statman \cite{statman79}:
\begin{theorem}\label{statmanstheorem}
Equality in the typed lambda calculus is not elementary recursive.
\end{theorem}
This has the consequence that the question whether $a \in \sigma$ is derivable or not, is decidable, but not feasibly so.

It should be said that from a philosophical point of view, this is slightly odd: this says that in most versions of type theory it may be infeasible to decide whether $a$ should count as compelling evidence in favour of $\sigma$. This means that one may be presented with evidence in favour of $\sigma$, but the task of deciding whether this is indeed conclusive evidence may be, from a practical point of view, impossible. This raises the question to which extent it should still count as \emph{compelling} evidence.

Again, this is more than just a philosophical problem. In an interesting paper \cite{geuverswiedijk08},  Geuvers and Wiedijk write:
\begin{quote}
  In theorem provers based on type theory the main performance bottleneck is the convertibility check: if the calculated type of a term $M$ is $A$, but it used in a context where the type should be $B$, then the system needs to verify that $A =_{\beta\iota\delta} B$, where $\delta$ is the equality arising from definitional expansion (unfolding definitions) and $\iota$ is the equality arising from functions defined by (higher order primitive) recursion. In fact, the inefficiency of the convertibility check means that type correctness is in practice only semi-decidable. Although in theory it is decidable whether a term $M$ has type $A$, in practice when it is not correct the system could be endlessly reducing and would not terminate in an acceptable time any more.
\end{quote}
To see what the problem is, let us consider the following question: is
\[ \vdash {\bf refl}(A(3,2^{65536} -3)) \in {\rm Id}(\mathbb{N}, A(3,2^{65536} -3), A(4,3)) \]
derivable, where $A$ is the Ackermann function?\footnote{This slightly modifies an example mentioned in Geuvers and Wiedijk.} Most type checkers would try to answer this question by normalising both $A(3,2^{65536} -3)$ and $A(4,3)$, and will not reach an answer before the end of the physical universe.

As Geuvers and Wiedijk comment, this shows how different the proofs terms in type theory are from what most mathematicians would imagine proofs to be. First of all, we would expect to be able to just read a proof (``follow it with our finger'' as Geuvers and Wiedijk write) and then become convinced of the truth of some mathematical statement. Secondly, we would expect proofs to have some explanatory function: it would tell us \emph{why} statements are true. If we take the term ${\bf refl}(A(3,2^{65536} -3))$ above, this hardly explains why $A(3,2^{65536} -3) = A(4,3)$: it basically just says that it is true.

Motivated by these questions, Geuvers and Wiedijk present a modification of type theory in which type checking is efficiently decidable. Their idea is that all these difficulties come from the fact that to do type checking a computer has to do many conversions, none of which are stored in the proof term. And, indeed, they show that if we enrich our proof terms with explicit conversions, we are able to make type checking feasible.

The purpose of this paper is to present a different implementation of their idea. In our view, our method has three advantages over theirs. First of all, it is systematic in that the way we store conversions in the proof terms relies on one simple idea. Secondly, our solution makes use of insights coming from homotopy type theory, so in that way it connects with a lot of exciting research which is currently happening. Thirdly, our proof of feasibility makes more realistic assumptions. For instance, Geuvers and Wiedijk assume that checking syntactic equality of strings can be done in unit time, while we assume that this requires a time linear in the length of the smallest string.

Our starting point is the observation that Martin-L\"of type theory has two kinds of equality: judgemental (or definitional) equality and propositional equality. Recent developments in HoTT strongly suggest that it should be possible to eliminate definitional equality completely in favour of propositional equality, and, indeed, that is what we do here. The trick is to state every computation rule in Martin-L\"of type theory (which is naturally thought of as a conversion) as a propositional equality, including the computation rule for the identity type itself. The first main result of this paper is that in doing this we obtain a system in which type checking is efficiently decidable: indeed, it is decidable in quadratic time.

Clearly, such a system is weaker than the standard systems, so it is natural to wonder how much weaker it is. We claim that this system still suffices for doing all of constructive mathematics; indeed, we conjecture that in such a system it should be possible to formalise most of the HoTT book. Of course, a detailed verification of such a claim would be very time and paper consuming, and we will not attempt that here. But it is generally understood (and this is also backed up by the success of cubical type theory) that having the computation rule for the identity type as a propositional equality is not an obstacle to doing homotopy type theory (see also papers by Coquand and Danielsson \cite{coquanddanielsson13} and Bocquet \cite{bocquet20}; we are also aware of a talk by Nicolai Kraus at TYPES-2017). As another indication, we will prove, and this will be our second main result, that the syntactic category of such a weak type theory is a \emph{path category with homotopy $\Pi$-types} (see \cite{bergmoerdijk18,vandenberg18,vandenberg20,denbesten20}). Given this work it should be clear that a lot of HoTT can be formalised in such a setting.

Let us finish by pointing out that there are other reasons why it is interesting to explore the consequences of abolishing definitional equality in favour propositional equality, which have nothing to do with feasibility. 

First of all, weaker rules have more models. As the work on cubical type theory makes clear, this is especially true from a constructive point of view.

Secondly, eliminating definitional equality helps us to understand it better. The idea here is that we do not fully understand what definitional equality is doing for us: to really understand this, the best way may be to try to live without it and see what happens. (For instance, an interesting question here is whether univalence still implies function extensionality in such a setting.)

Thirdly, it is hard to find objective grounds for deciding which equalities are definitional and which ones are only propositional. The way things are done usually, for an arbitrary $x \in \mathbb{N}$, we have that $x + 0 = x$ holds definitionally, while $0 + x = x$ holds only propositionally. In fact, it is possible to define a variant of plus for which the opposite holds. This is not just odd, but can also be akward in formalisations. Even Martin-L\"of himself suggested to the authors that such decisions are best made on pragmatic grounds. For the last reason the authors of this paper like to think of the type theory we are proposing as a kind of \emph{objective type theory}.\footnote{Here is an interesting quote in that respect: ``My personal disenchantment with dependent type theories coincides with the decision to shift from extensional to intensional equality. This meant for example that $0 + n = n$ and $n + 0 = n$ would henceforth be regarded as fundamentally different assertions, one an identity holding by definition and the other a mere equality proved by induction. Of course I was personally upset to see several years of work, along with Constable's Nuprl project, suddenly put beyond the pale. But I also had the feeling that this decision had been imposed on the community rather than arising from rational discussion. And I see the entire homotopy type theory effort as an attempt to make equality reasonable again.'' \cite{paulson18}} 

The contents of this paper are therefore as follows. In Section 2 we present objective type theory. In Section 3 we prove that type checking in objective type theory can be done in quadratic time. In Sections 4 -- 6 we prove that the classifying category associated to objective type theory is a path category with homotopy $\Pi$-types. We do this in several steps: we construct the classifying category in Section 4, show that it is a path category in Section 5 and construct the homotopy $\Pi$-types in Section 6. We end with some directions for future research in Section 7.

\section{Objective type theory}

The goal of this section will be to introduce objective type theory. 

\begin{remark}\label{somesyntacticconventions}
We imagine that this type theory has been formulated using some device like De Bruijn-indices so that there is no difference between $\alpha$-equivalence and syntactic equality of expressions. However, for the sake of human readability, we will be using variables. But this means that the notation $s \equiv t$, which we will use for syntactic equality of expressions, will act like $\alpha$-equivalence.
\end{remark}

\subsection{Basic judgements and rules} The type theory we wish to introduce derives statements of one of the following three forms:
\begin{displaymath}
\begin{array}{lll} \vdash \Gamma \, \mbox{Ctxt} & \Gamma \vdash \sigma \, \in \, \mbox{Type} & \Gamma \vdash a \in \sigma
\end{array}
\end{displaymath}
The meaning of the first statement is that $\Gamma$ is a \emph{context}, that of the second is that $\sigma$ is a \emph{type} in context $\Gamma$, whilst the third means that $a$ is a \emph{term} of type $\sigma$ in context $\Gamma$.

Our type theory has only three basic rules. The first two tell us how to form contexts, the idea being that context $\Gamma$ is a list of the form
\[  [ \, x_1 : \sigma_1, x_2 : \sigma_2, \ldots, x_{n-1} : \sigma_{n-1} \, ] \]
with the $x_i$ being distinct variables. In that spirit, the two context formation rules are:
\begin{center}
\begin{tabular}{cc}
\begin{tabular}{c}
\\
\hline
$\vdash [] \, \mbox{Ctxt}$
\end{tabular}
&
\begin{tabular}{c}
$\vdash \Gamma \, \mbox{Ctxt} \qquad \Gamma \vdash \sigma \, \in \, 			\mbox{Type} \qquad x \, \mbox{ fresh}$ \\
\hline
$\vdash [\Gamma, x:\sigma] \, \mbox{Ctxt}$
\end{tabular}
\end{tabular}
\end{center}

The third and final basic rule of our type theory is the following variable rule, which says that if $x \in \sigma$ occurs in a context $\Gamma$, then $\Gamma \vdash x \in \sigma$ holds. Formally:
\begin{center}
\begin{tabular}{c}
$\vdash [\Gamma, x \in \sigma, \Delta] \, \mbox{Ctxt} $ \\
\hline
$\Gamma, x \in \sigma, \Delta \vdash x \in \sigma $
\end{tabular}
\end{center}

\subsection{The syntax} The other rules follow the usual pattern of Martin-L\"of type theory, in which each type constructor comes with four rules (formation, introduction, elimination and computation). In the appendix the reader can find the rules for identity and $\Pi$-types. They are the usual ones, except that the computation rule holds only \emph{propositionally}: that is, it states a propositional instead of a definitional equality. This means that if we think of type theory as some sort of generalised algebraic theory, there will only be constructors, but no equations. In the formulation of these rules, we have used $[t_1,\ldots,t_n/x_1,\ldots,x_n]$ for the result of a (capture avoiding) substitution of $t_1,\ldots,t_n$ for $x_1,\ldots,x_n$, respectively. In addition, $[x_1,\ldots,x_n]t$ means that the variables $x_1,\ldots,x_n$ have become bound in the term or type $t$.

We emphasise that we think of objective type theory as an open framework which need not only include the $\Pi$- and Id-types, but could also include rules for $\Sigma$-types and a natural numbers type, for instance. Again, we imagine that these would be formulated in the usual way but with the computation rules in the form of propositional equalities.

\begin{remark}\label{contextualparameter}
In the appendix the reader can also find strengthened versions of these rules with a general contextual parameter. An important step in the development below will be the proof that these are admissable.
\end{remark}

\subsection{Some simple properties of objective type theory} We start off by making a number of simple observations about our axiomatisation of objective type theory.

\begin{lemma}\label{admissabilityofweakeningandsubst} {\rm (Admissability of weakening and substitution)}
The following weakening and substitution rules are admissable in the system:
\begin{center}
\begin{tabular}{cc}
\begin{tabular}{c}
$\Gamma, \Delta \vdash \mathcal{J} \qquad \Gamma \vdash \sigma \, \mbox{\emph{type}} $\\
\hline
$\Gamma, x \in \sigma, \Delta \vdash \mathcal{J} $
\end{tabular}
&
\begin{tabular}{c}
$\Gamma \vdash a \in  \sigma  \qquad \Gamma, x \in \sigma, \Delta \vdash \mathcal{J}$\\
\hline
$\Gamma, \Delta[a/x] \vdash \mathcal{J}[a/x] $
\end{tabular}
\end{tabular}
\end{center}
\end{lemma}

\begin{proof}
By induction on the derivations of $\Gamma, \Delta \vdash \mathcal{J} $ and $\Gamma, x \in \sigma, \Delta \vdash \mathcal{J}$, respectively.
\end{proof}

\begin{lemma}\label{uniquederivability}
 Each derivable judgement has a unique derivation.
\end{lemma}
\begin{proof}
 The reason is that each judgement, whether it is $\vdash \Gamma \, \mbox{Ctxt}, \Gamma \vdash \sigma \in \mbox{Type}$ or $\Gamma \vdash a \in \sigma$ appears as the conclusion of at most rule in our system. Indeed, this is clear for contexts, because the first two basic rules are the only ones introducing contexts into the system. For judgements of the form $\Gamma \vdash \sigma \in {\rm Type}$ we look at the main type constructor in $\sigma$: the introduction rule for that type constructor is the only one introducing judgements of that shape into the system. Similarly, for judgements of the form $\Gamma \vdash a \in \sigma$ we look at the main term constructor in $a$: the rule introducing that constructor into the system is the only one which has a conclusion with that precise shape.
\end{proof}

In a similar fashion one shows: 
\begin{lemma}\label{uniquenessoftypes} {\rm (Uniquess of types)}
If $\Gamma \vdash a \in \sigma$ and $\Gamma \vdash a \in \tau$, then $\sigma \equiv \tau$.
\end{lemma}

\section{Quadratic type checking}

We now come to the first main result of this paper.

\begin{theorem}\label{maintheorem}
The question whether $\Gamma \vdash a \in \sigma$ is derivable or not can be decided in quadratic time. 
\end{theorem}
  
\begin{proof}
When estimating the time it takes to decide $\Gamma \vdash a \in \sigma$, we may always assume that this judgement is actually derivable. For if it is not and our decision procedure exceeds this time estimate without having reached a decision, we will simply conclude that $\Gamma \vdash a \in \sigma$ must not have been derivable, time-out the computation and output \texttt{false}. Write $\Gamma \vdash^{*} a \in \sigma$ for the promise problem of deciding $\Gamma \vdash a \in \sigma$ with promise $\Gamma \vdash \sigma \in {\rm Type}$. In other words, in the starred version of the problem it is permitted to give a wrong answer whenever $\Gamma \nvdash \sigma \in {\rm Type}$. Similarly, write $\Gamma \vdash^{*} \sigma \in {\rm Type}$ for the promise problem of deciding $\Gamma \vdash \sigma \in {\rm Type}$ with promise $\vdash \Gamma \, {\rm Ctxt}$. Denote the length of a string $s$ by $\lvert s \rvert$. It is sufficient to prove that $\Gamma \vdash^{*} a \in \sigma$ and $\Gamma \vdash^{*} \sigma \in {\rm Type}$ can be decided in time $\mathcal{O}((\lvert a \rvert + \lvert \sigma \rvert)^{2})$ and $\mathcal{O}(\lvert \sigma \rvert^{2})$ respectively, for the following reason. One easily verifies that a procedure to decide $\Gamma \vdash^{*} \sigma \in {\rm Type}$ in time $\mathcal{O}(\lvert \sigma \rvert^{2})$ gives rise to a procedure to decide $\vdash \Gamma \in {\rm Ctxt}$ in time $\mathcal{O}(\lvert \Gamma \rvert^{2})$, by induction on $\lvert \Gamma \rvert$. Deciding $\Gamma \vdash a \in \sigma$ is then just a matter of deciding $\vdash \Gamma \, {\rm Ctxt}$, $\Gamma \vdash^{*} \sigma \in {\rm Type}$ and $\Gamma \vdash^{*} a \in \sigma$.
  
 We prove that $\Gamma \vdash^{*} a \in \sigma$ and $\Gamma \vdash^{*} \sigma \in {\rm Type}$ can be decided in time $k_{2} \lvert a \rvert^{2} + k_{1} \lvert \sigma \rvert$ and $k_{2} \lvert \sigma \rvert^{2}$ respectively, for sufficiently large constants $k_{1} < k_{2}$, by (simultaneous) induction on $\lvert a \rvert$ in the case of $\Gamma \vdash^{*} a \in \sigma$ and $\lvert \sigma \rvert$ in the case of $\Gamma \vdash^{*} \sigma \in {\rm Type}$. The base case, $\Gamma \vdash^{*} x \in \sigma$, can clearly be decided in $k_{2} + k_{1} \lvert \sigma \rvert$ time steps, by comparing the string $\sigma$ to the string appearing in $\Gamma$ at the position pointed to by $x$. We will go through the induction step for each rule of the $\Pi$-type and Id-type.

 We start with the rules for the $\Pi$-type.
  
 \emph{Formation.} Let the problem $\Gamma \vdash^{*} \Pi(A,[x]B)  \in {\rm Type}$ be given. We make two recursive calls to the algorithm to decide $\Gamma \vdash^{*} A \in {\rm Type}$ and $\Gamma, x \in A \vdash^{*} B \in {\rm Type}$, at the cost of $k_{2} (\lvert A \rvert^{2} + \lvert B \rvert^{2})$ time steps. This is clearly less than $k_{2} (\lvert A \rvert + \lvert B \rvert + 1)^{2}$, which is the number of time steps available to us.
  
 \emph{Introduction.} Let the problem $\Gamma \vdash^{*} \lambda(A, [x]B, [x]t) \in \sigma$ be given. We make the string comparison $\texttt{equal}(\Pi(A,[x]B),\sigma)$ at the cost of $k_1|\sigma|$ time steps. Moreover, we make a recursive call to the algorithm to decide $\Gamma, x \in A \vdash^{*} t \in B$, at the cost of $k_{2} \lvert t \rvert^{2} + k_{1} \lvert B \rvert$ additional time steps. As $k_1 < k_2$, we have clearly used less than $k_{2} (\lvert t \rvert + |A| + |B| + 1)^{2} + k_{1}\lvert \sigma \rvert$ time steps, which is the number of time steps available to us. There is no need to check the premises $\Gamma \vdash A \in {\rm Type}$ and $\Gamma, x \in A \vdash B \in {\rm Type}$. Since we are working under the assumption that $\Gamma \vdash \Pi(A,[x]B) \in {\rm Type}$ is derivable, it follows from unique derivability that $\Gamma \vdash A \in {\rm Type}$ and $\Gamma, x \in A \vdash B \in {\rm Type}$ are derivable.
  
 \emph{Elimination.} Let the problem $\Gamma \vdash^{*} {\bf app}(A,[x]B,f,a) \in \sigma$ be given. We make the string comparison $\texttt{equals}(B[a/x], \sigma)$ at the cost of $k_{1} \lvert \sigma \rvert$ time steps. Moreover, we make four recursive calls to the algorithm to decide $\Gamma \vdash^{*} A \in {\rm Type}$; $\Gamma, x \in A \vdash^{*} B \in {\rm Type}$; $\Gamma \vdash^{*} f \in \Pi(A,[x]B)$ and $\Gamma \vdash^{*} a \in A$, at the cost of $k_{2}( \lvert A \rvert^{2} + \lvert B \rvert^{2} + \lvert f \rvert^{2} + \lvert a \rvert^{2}) + k_{1}(2 \lvert A \rvert + \lvert B \rvert + 1)$. Since we may safely assume that $k_{2} \geq 2 k_{1}$, one easily verifies that we have used less than $k_{2}( \lvert A \rvert + \lvert B \rvert + \lvert f \rvert + \lvert a \rvert + 1)^{2} + k_{1} \lvert \sigma \rvert$ time steps in total, which is the number of time steps available to us.
  
 \emph{Computation.} Let the problem $\Gamma \vdash^{*} {\bf betaconv}(A,[x]B,a,[x]t) \in \sigma$ be given. We make the string comparison $\texttt{equals}({\bf app}(A,[x]B,\lambda(A,[x]B,[x]t),a) =_{B[a/x]} t[a/x], \sigma)$ at the cost of $k_{1} \lvert \sigma \rvert$ time steps. This is clearly less than $k_{2} (\lvert A \rvert + \lvert B \rvert + \lvert a \rvert + \lvert t \rvert +  1)^{2} + k_{1} \lvert \sigma \rvert$, which is the number of time steps available to us. There is no need to check the premises $\Gamma, x \in A \vdash t \in B$ and $\Gamma \vdash a \in A$. Since we are working under the assumption that $\Gamma \vdash \sigma \in {\rm Type}$ is derivable and we have verified that $\sigma \equiv {\bf app}(A,[x]B,\lambda(A,[x]B,[x]t),a) =_{B[a/x]} t[a/x]$, it follows from unique derivability that $\Gamma, x \in A \vdash t \in B$ and $\Gamma \vdash a \in A$ are derivable.

Finally, we check the rules for the Id-type.

 \emph{Formation.} Let the problem $\Gamma \vdash^{*} a =_{A} b  \in {\rm Type}$ be given. We make three recursive calls to the algorithm to decide $\Gamma \vdash^{*} A \in {\rm Type}$, $\Gamma \vdash^{*} a \in A$ and $\Gamma \vdash^{*} b \in A$, at the cost of $k_{2} (\lvert A \rvert^{2} + \lvert a \rvert^{2} + \lvert b \rvert^{2})+ 2 k_1 |A|$ time steps. Since we may assume that $k_2 \geq 2k_1$, this is clearly less than $k_{2} (\lvert A \rvert + \lvert a \rvert + \lvert b \rvert + 1)^{2}$, which is the number of time steps available to us.
  
 \emph{Introduction.} Let the problem $\Gamma \vdash^{*} {\bf refl}(A,a) \in \sigma$ be given. We make the string comparison $\texttt{equals}(a =_{A} a, \sigma)$ at the cost of $k_{1} \lvert \sigma \rvert$ time steps. Moreover, we make a recursive call to the algorithm to decide $\Gamma \vdash^{*} a \in A$ at the cost of $k_{2} \lvert a \rvert^{2} + k_{1} \lvert A \rvert$ time steps. Since $k_{1} \leq k_{2}$, it is clear that we have used less than $k_{2}( \lvert A \rvert + \lvert a \rvert + 1)^{2} + k_{1} \lvert \sigma \rvert$ time steps in total, which is the number of time steps available to us. There is no need to check $\Gamma \vdash A \in {\rm Type}$. Since we are working under the assumption that $\Gamma \vdash \sigma \in {\rm Type}$ is derivable and we have verified that $\sigma \equiv a =_{A} a$, it follows from unique derivability that $\Gamma \vdash A \in {\rm Type}$ is derivable.
  
 \emph{Elimination.} Let the problem $\Gamma \vdash^{*} {\bf idrec}(A, [x,y,u]P,a,b,p,[x]d)) \in \sigma$ be given. We make the string comparison $\texttt{equals}(P[a,b,p /x,y,u], \sigma)$ at the cost of $k_{1} \lvert \sigma \rvert$ time steps. Moreover, we make six recursive calls to the algorithm to decide $\Gamma \vdash^{*} A \in {\rm Type}$; $\Gamma, x \in A, y \in A, u \in x =_{A} y \vdash^{*} P \in {\rm Type}$; $\Gamma \vdash^{*} a \in A$; $\Gamma \vdash^{*} b \in A$; $\Gamma \vdash^{*} p \in a =_{A} b$ and $\Gamma, x \in A \vdash^{*} d \in P[x,x,{\bf refl}(A,x)/x,y,u]$. Using the estimate $\lvert P[x,x,{\bf refl}(A,x)/x,y,u] \rvert \leq \lvert P \rvert ( \lvert A \rvert + 2 )$, these calls cost at most $k_{2}( \lvert A \rvert^{2} + \lvert P \rvert^{2} + \lvert a \rvert^{2} + \lvert b \rvert^{2} + \lvert p \rvert^{2} + \lvert d \rvert^{2}) + k_{1}(3 \lvert A \rvert + \lvert a \rvert + \lvert b \rvert + \lvert P \rvert ( \lvert A \rvert + 2 ) + 1)$. Since we may safely assume that $k_{2} \geq 3 k_{1}$, one easily verifies that we have used less than $k_{2}( \lvert A \rvert + \lvert P \rvert + \lvert a \rvert + \lvert b \rvert + \lvert p \rvert + \lvert d \rvert + 1)^{2} + k_{1} \lvert \sigma \rvert$ time steps in total, which is the number of time steps available to us.
  
 \emph{Computation.} Let the problem $\Gamma \vdash^{*} {\bf idconv}(A, [x,y,u]P,a,[x]d) \in \sigma$ be given. We make the string comparison $\texttt{equals}({\bf idrec}(A,[x,y,u]P,a,a,{\bf refl}(A,a),[x]d) =_{P[a,a,{\bf refl}(A,a)/x,y,u]} d[a/x], \sigma)$ at the cost of $k_{1} \lvert \sigma \rvert$ time steps. This is clearly less than $k_{2} (\lvert A \rvert + \lvert P \rvert + \lvert a \rvert + \lvert d \rvert + 1)^{2} + k_{1} \lvert \sigma \rvert$, which is the number of time steps available to us. There is no need to check any of the premises. Since we are working under the assumption that $\Gamma \vdash \sigma \in {\rm Type}$ is derivable and we have verified that $\sigma \equiv {\bf idrec}(A,[x,y,u]P,a,a,{\bf refl}(A,a),[x]d) =_{P[a,a,{\bf refl}(A,a)/x,y,u]} d[a/x]$, it follows from unique derivability that all of the premises are derivable.
\end{proof}

\begin{remark} \label{othertypeconstructors}
  As we mentioned before, we think of objective type theory as an open framework which can be extended with additional rules like those for a sum and a natural numbers type. If we extended objective type theory with these type constructors, the proof that we gave above would still work; but we will not attempt to formulate and prove a general statement for general type constructors.
\end{remark}

\section{The classifying category}

The second main result of this paper will be that the classifying category associated to propositional type theory is a path category with weak homotopy $\Pi$-types in the sense of \cite{bergmoerdijk18}. In fact, if we assume function extensionality the homotopy $\Pi$-types will be strong in the sense of \cite{denbesten20}. We will prove this result in several steps. In this section we will outline the construction of the classifying category. This is included for the convenience of the reader, because the construction is absolutely standard (see, for instance, the account in \cite[Section 2.4]{hofmann97} which we will closely follow here). In the next we will prove that it is a path category; in the one after that we will construct weak and strong homotopy $\Pi$-types.

\subsection{Generalised substitution}  A \emph{context morphism} $f: \Delta \to \Gamma$ with
\[ \Gamma =  [ \, x_1 \in \sigma_1, x_2 \in  \sigma_2, \ldots, x_n \in \sigma_n \, ]. \]
is a sequence of terms $(t_1, \ldots, t_{n})$ for which the following statements are derivable:
\begin{displaymath}
\begin{array}{l}
\Delta \vdash t_1 \in \sigma_1 \\ \Delta \vdash t_2 \in \sigma_2[t_1/x_1]  \\ \ldots \\ \Delta \vdash t_n \in \sigma_n[t_1,\ldots,t_{n-1}/x_1,\ldots,x_{n-1}]
\end{array}
\end{displaymath}
We will regard two context morphisms $(s_1,\ldots,s_n), (t_1,\ldots,t_n): \Delta \to \Gamma$ as (syntactically) equal if they are componentwise syntactically equal.

Note that for a context $\Gamma$ as above we always have a context morphism $\Gamma \to \Gamma$:
\[ 1_\Gamma :\equiv (x_1, \ldots, x_n). \]

Context morphisms allow a notion of \emph{generalised substitution}. Indeed, if $\theta$ is a type, term or judgement in context $\Gamma$ and $f: \Delta \to \Gamma$ is a context morphism as above, then there is a new type, term or judgement
\[ \theta[f] :\equiv \theta[t_1,\ldots,t_n/x_1,\ldots,x_n] \]
in context $\Delta$. We clearly have $\theta[1_\Gamma] \equiv \theta$.

With this notion of generalised substitution, we can also regard the context morphisms as being inductively defined by the following rules:
\begin{center}
\begin{tabular}{cc}
\begin{tabular}{c}
$\Delta \, \mbox{Ctxt}$ \\
\hline
$!_\Delta: \Delta \to []$
\end{tabular}
&
\begin{tabular}{c}
$f: \Delta \to \Gamma \qquad \Gamma \vdash \sigma \, \in \mbox{Type} \qquad \Delta \vdash t : \sigma[f]$ \\
\hline
$(f, t): \Delta \to [\, \Gamma, x : \sigma \, ]$
\end{tabular}
\end{tabular}
\end{center}
We also have the following generalising substitution rule:
\begin{proposition}\label{gensubstrule}
If $f: \Delta \to \Gamma$ is a context morphism and $\Gamma, \Theta \vdash \mathcal{J}$, then also $\Delta, \Theta[f] \vdash \mathcal{J}[f]$.
\end{proposition}
\begin{proof} This is Proposition 2.12 in \cite{hofmann97}. By induction on the structure of $f$. The base case follows from weakening and in the induction step we use the induction hypothesis and the substitution rule.
\end{proof}

\subsection{Classifying category} If $g: \Theta \to \Delta$ and $f: \Delta \to \Gamma$ with $f \equiv (t_0,\ldots,t_n)$, then we can form a new list of terms:
\[ f \circ g :\equiv (t_0[g], \ldots, t_n[g]). \]
Then $f \circ g$ is a context morphism $\Theta \to \Gamma$. In fact, we have:
\begin{proposition}\label{generalitiesaboutgensubst}
If $h: E \to \Theta, g: \Theta \to \Delta, f: \Delta \to \Gamma$, and $\Gamma \vdash \sigma \in \mbox{Type}$ and $\Gamma \vdash a \in \sigma$, then the following equations hold up to syntactic equality:
\begin{eqnarray*}
1_\Gamma \circ f & \equiv & f \equiv f \circ 1_\Delta \\
f \circ (g \circ h) & \equiv & (f \circ g) \circ h \\
\sigma[f \circ g] & \equiv & \sigma[f][g] \\
a[f \circ g] & \equiv & a[f][g]
\end{eqnarray*}
\end{proposition}
\begin{proof}
See \cite[Proposition 2.13]{hofmann97}.
\end{proof}

The first two items of this proposition tell us that we have indeed defined a category. Within this category, two classes of morphisms will become important. 

\begin{definition}\label{fibrationsanddisplaymaps}
  {\rm If
  \[ \Gamma \vdash \sigma \in \mbox{Type}, \]
  then there is a map of the form $[\Gamma, x \in \sigma] \to \Gamma$ dropping the last type from the context (more precisely, if $\Gamma = [x_1 \in \sigma_1,\ldots,x_{n-1} \in \sigma_{n-1}]$, then it is the sequence $(x_1,\ldots,x_{n-1})$). We will call maps of this form \emph{display maps}. Note that a section of this display map corresponds to a term $a$ of type $\sigma$ in context $\Gamma$. By closing the class of display maps under isomorphisms and composition, we obtain the class of \emph{fibrations}.}
\end{definition}

An important property of display maps and fibrations is that they are both stable under pullback. The reason is that if $\Gamma \vdash a \in \sigma$ and $f: \Delta \to \Gamma$ is a context morphism, then
\[
\begin{tikzcd}
{[\Delta, y \in \sigma[f] ]} \arrow[r, "{[f,y]}"] \arrow[d] & {[\Gamma, x \in \sigma]} \arrow[d] \\
\Delta \arrow[r, swap, "f"] & \Gamma
\end{tikzcd}
\]
is a pullback (see, for instance, \cite[Proposition 3.9]{hofmann97}).

\section{Classifying category is a path category}

Having constructed the classifying category, we will now show that it has the structure of a path category (see \cite{bergmoerdijk18} for the definition) whose fibrations are precisely those maps which we have baptised fibrations in \ref{fibrationsanddisplaymaps}. To make the proof less cumbersome, we will often write $x =_A y$ (or $x = y$) to mean that the type $x =_A y$ is inhabited (as in the HoTT book, for instance).

\begin{lemma}\label{basicpropofeq}
  Propositional equality is a congruence.
\end{lemma}
\begin{proof}
  The argument is absolutely standard (see, for instance, \cite[Chapter 2]{UFP13}). We only need to verify that the standard proof can be performed in our objective type theory. 

  First of all, we have refl to witness that equality is reflexive.

  Secondly, suppose $x \in A \vdash B \in \mbox{Type}$ and $\vdash p \in a =_A b$. Then from $B[a/x]$ being inhabited we wish to deduce that $B[b/x]$ is inhabited as well (the Leibniz principle). Write $Q :\equiv \Pi(B,[z]B[y/x])$ and suppose that $t \in B[a/x]$. From
  \[ x \in A, y \in A, u \in x=_A y \vdash Q \in \mbox{Type} \]
  and
  \[ x \in A \vdash \lambda(B,[z]B,[z]z) \in Q[x,x,{\bf refl}(A,x)/x,y,u] \]
  it follows that
  \[ \Pi(B[a/x],[z]B[b/x]) \]
  is inhabited, by $s$ say. So if $t \in B[a/x]$, then
  \[ {\bf app}(B[a/x],[z]B[b/x], s, t) \in B[b/x], \]
  as desired.

  From this symmetry and transitivity follow as well. Indeed, choose $P :\equiv x =_A a$. Since ${\bf refl}(A,a) \in P[a/x]$, any $p \in a=_A b$ gives that $P[b/x] \equiv b =_A a$ is inhabited as well.

  Furthermore, if $p \in a=_A b$ and $q \in b=_A c$, consider $P :\equiv x =_A c$. Since $P[b/x]$ is inhabited, so must be $P[a/x]$, using symmetry and the Leibniz principle.
\end{proof}

\begin{lemma}\label{table4admissable}
  One can define appropriate terms and types satisfying the rules for contextual $\Pi$-types (the rules in Table \ref{table4} in the appendix).
\end{lemma}
\begin{proof}
  By induction on the length of the context $\Delta$.

  If $\Delta = []$, then we can put \[ \Pi([],[]B) :\equiv B, \]
  and
  \[ \lambda([],[]B,[]t) :\equiv t, \]
  while
  \[ {\bf app}([],[]B,f,[]) :\equiv f. \]
  Then
  \[ {\bf app}([],[]B,\lambda([],[]B,[]t),[]) \equiv \lambda([]t) \equiv t \equiv t[[]/[]], \]
  so we can witness {\bf betaconv} by {\bf refl}.

  If the statement is true for $\Delta$, then we can prove it for $ \Delta.A = [\Delta, y \in A]$, as follows. We put \[ \Pi(\Delta.A, [\vec{x}]B) :\equiv \Pi(\Delta,[\vec{y}]\Pi(A,[y]B)), \]
  and
  \[ \lambda(\Delta.A,[\vec{x}]B,[ \vec{x}]t) :\equiv \lambda(\Delta,[\vec{y}]\Pi(A,[y]B),[\vec{y}]\lambda(A,[y]B,[y]t)), \]
  while
  \[ {\bf app}(\Delta.A,[\vec{x}]B,f,(\vec{a},a)) :\equiv {\bf app}(A,[y ]B,{\bf app}(\Delta,[\vec{y}]\Pi(A,[y]B),f,\vec{a}),a). \]
  Then we have:
  \begin{eqnarray*}
    {\bf app}(\Delta.A,[\vec{x}]B,\lambda(\Delta.A,[\vec{x}]B,[\vec{x}]t),(\vec{a},a)) & \equiv \\
    {\bf app}(A,[y]B,{\bf app}(\Delta,[\vec{y}]\Pi(A,[y]B),\lambda(\Delta,[\vec{y}]\Pi(A,[y]B),[\vec{y}]\lambda(A,[y]B,[y]t)),\vec{a}),a) & = \\
    {\bf app}(A,[y]B,\lambda(A,[y]B, [y]t)[\vec{a}/\vec{y}],a) & = \\
    t[(\vec{a},a)/\vec{x}],
  \end{eqnarray*}
  using both the induction hypothesis and the previous lemma.
\end{proof}

\begin{lemma}\label{table5admissable}
  One can define appropriate terms and types satisfying the rules for contextual identity types (that is, the rules in Table \ref{table5} in the appendix).
\end{lemma}
\begin{proof}
Suppose
\begin{center}
\begin{tabular}{c}
$\Gamma, x \in A, y \in A, u \in x=_A y, \Delta \, \vdash P \in \mbox{Type}$ \\
$\Gamma, x \in A, \vec{z} \in \Delta[x,x,{\bf refl}(A,x)/x,y,u] \, \vdash d \in P[x, x, {\bf refl}(A,x)/x,y,u]$ \\
$\Gamma \, \vdash p \in a=_A b$ \\
$\Gamma \, \vdash \vec{q} \in \Delta[a,b,p/x,y,u]$
\end{tabular}
\end{center}
Writing 
\begin{eqnarray*} Q & :\equiv & \Pi(\Delta, [\vec{z}]P) \quad \mbox{ and } \\ s & :\equiv & \lambda(\Delta,[\vec{z}]P, [\vec{z}] d)[x,x,{\bf refl}(A,x)/x,y,u],
\end{eqnarray*}
 we have:
\begin{center}
\begin{tabular}{c}
$\Gamma, x \in A, y \in A, u \in x=_A y \, \vdash Q \in \mbox{Type}$ \\
$\Gamma, x \in A \, \vdash s \in Q[x, x, {\bf refl}(A,x)/x,y,u]$ \\
$\Gamma \, \vdash p \in a=_A b$
\end{tabular}
\end{center}
so we have a term
\[ {\bf idrec}(A,[x,y,u]Q,a,b,p,[x]s) \]
in $Q[a,b,p/x,y,u]$ and therefore a term
\[ {\bf app}(\Delta,[\vec{z}]P,{\bf idrec}(A,[x,y,u]Q,a,b,p,[x]s),\vec{q})[a,b,p/x,y,u] \]
in $P[a,b,p/x,y,u]$, as desired.

To complete the proof, we assume $a \in A$ and $\vec{w} \in \Delta[a,a,{\bf refl}(A, a)/x,y,u]$ and we have to calculate:
\begin{eqnarray*}
  {\bf app}(\Delta,[\vec{z}]P,{\bf idrec}(A,[x,y,u]Q,a,a,{\bf refl}(A,x),[x]s),\vec{w})[a,a,{\bf refl}(A,a)/x,y,u] & = \\
  {\bf app}(\Delta,[\vec{z}]P,\lambda(\Delta,[\vec{z}]P, [\vec{z}] d[a/x]),\vec{w})[a,a,{\bf refl}(A,a)/x,y,u] & = \\
  d[a,\vec{w}/x,\vec{z}].
\end{eqnarray*}
\end{proof}

\begin{theorem}\label{syntacticcatpathcat}
  The syntactic category associated to our type theory is a path category.
\end{theorem}
\begin{proof}
  This follows from the previous lemma and the main result of \cite{vandenberg18}.
\end{proof}

\section{Classifying category has homotopy $\Pi$-types}

We will now prove that the classifying category is a path category with homotopy $\Pi$-types (see \cite{denbesten20} for the definition). We will give a categorical proof heavily exploiting the results from \cite{denbesten20}.

\begin{definition}\label{classofdisplaywithpi}
  {\rm We say that a class of fibrations $\mathcal{D}$ in a path category is a \emph{class of display maps with weak (strong) homotopy $\Pi$-types} if the following conditions are satisfied:
  \begin{itemize}
  \item{Every identity map lies in $\mathcal{D}$.}
  \item{The pullback of a map in $\mathcal{D}$ along any other map can be found in $\mathcal{D}$.}
  \item{For every composable $d, e \in \mathcal{D}$ the weak (strong) homotopy $\Pi$-type $\Pi_{e} (d)$ exists and can be found in $\mathcal{D}$.}
  \end{itemize} }
\end{definition}
  
\begin{lemma}\label{pidisplaysyntacticcat}
  The display maps in the classifying category form a class of display maps weak homotopy $\Pi$-types, which will be strong as soon as function extensionality holds.
\end{lemma}  
\begin{proof}
This is a standard exercise in translating type-theoretic definitions in categorical terms. Note that by function extensionality we mean Axiom 2.9.3 from the HoTT book \cite{UFP13}.
\end{proof}

  \begin{lemma}\label{addiso}
  Let $\mathcal{D}$ be a class of display maps with weak (strong) homotopy $\Pi$-types in a path category $\mathcal{C}$. Write $\mathcal{I}$ for the class of isomorphisms in $\mathcal{C}$. Then $\mathcal{D} \cup \mathcal{I}$ is a class of display maps with weak (strong) homotopy $\Pi$-types as well.
  \end{lemma}
  
  \begin{proof}
  Let $d : X \longrightarrow I$ be a map in $\mathcal{D}$ and let $i : I \longrightarrow J$ be a map in $\mathcal{I}$. We claim that the weak (strong) homotopy $\Pi$-type $\Pi_{d} (i)$ exists and can be found in $\mathcal{D}$. It is not difficult to verify that $id : X \longrightarrow J$ is a weak (strong) homotopy $\Pi$-type for $d$ and $i$. Now note that $id : X \longrightarrow J$ is the pullback of $d$ along $i^{-1}$, hence isomorphic to some $d' : X' \longrightarrow J$ in $\mathcal{D}$, which is therefore a weak (strong) homotopy $\Pi$-type for $d$ and $i$ as well. If on the other hand $i : X \longrightarrow I$ is a map in $\mathcal{I}$ and $f : I \longrightarrow J$ is an arbitrary map, then it is easy to see that $1 : J \longrightarrow J$ is a weak (strong) homotopy $\Pi$-type for $i$ and $f$. The lemma follows since isomorphisms are fibrations and are closed under pullbacks.
  \end{proof}
  
    \begin{lemma}\label{horizontalpi}
  Let $\mathcal{D}$ be a class of display maps with weak (strong) homotopy $\Pi$-types in a path category $\mathcal{C}$. Write $\overline{\mathcal{D}}$ for the closure of $\mathcal{D}$ under composition. Then for every composable $d \in \mathcal{D}$ and $f \in \overline{\mathcal{D}}$, the weak (strong) homotopy $\Pi$-type $\Pi_{f}(d)$ exists.
  \end{lemma}
  
  \begin{proof}
  Consider maps $d : X \longrightarrow I$, $f : I \longrightarrow J$ and $e : J \longrightarrow K$, with $d,e \in \mathcal{D}$ and $f \in \overline{\mathcal{D}}$. By induction, it suffices to show that the weak (strong) homotopy $\Pi$-type $\Pi_{ef}(d)$ exists and can be found in $\mathcal{D}$, whenever the weak (strong) homotopy $\Pi$-type $\Pi_{f}(d)$ exists and can be found in $\mathcal{D}$. Take the weak (strong) homotopy $\Pi$-type $\Pi_{e} \Pi_{f}(d)$ and note that it can be found in $\mathcal{D}$. By Lemma 5.3 of \cite{denbesten20}, $\Pi_{e} \Pi_{f}(d)$ is a weak (strong) homotopy $\Pi$-type for $d$ and $ef$.
  \end{proof}
  
    \begin{lemma}\label{verticalpi}
  Let $\mathcal{D}$ be a class of display maps with weak (strong) homotopy $\Pi$-types in a path category $\mathcal{C}$. Write $\overline{\mathcal{D}}$ for the closure of $\mathcal{D}$ under composition. Then for every composable $f, g \in \overline{\mathcal{D}}$, the weak (strong) homotopy $\Pi$-type $\Pi_{f}(g)$ exists.
  \end{lemma}
  
  \begin{proof}
  Consider maps $f : I \longrightarrow J$, $g : Y \longrightarrow I$ and $d : X \longrightarrow Y$, with $f,g \in \overline{\mathcal{D}}$ and $d \in \mathcal{D}$. By induction, it suffices to show that the weak (strong) homotopy $\Pi$-type $\Pi_{f}(gd)$ exists, whenever the weak (strong) homotopy $\Pi$-type $\Pi_{f}(g)$ exists. Take the pullback
  \begin{equation*}
  \begin{tikzcd}
  Q \arrow[r] \arrow[d, swap, "e"] & X \arrow[d, "d"] \\
  (\Pi_{f} (g)) \times_{J} I \arrow[r, swap, "\varepsilon_{Y}"] & Y
  \end{tikzcd}
  \end{equation*}
  such that $e$ lies in $\mathcal{D}$. By pullback pasting, the pullback of a map in $\overline{\mathcal{D}}$ along any other map can be found in $\overline{\mathcal{D}}$, so in particular $\pi_{1} : (\Pi_{f} (g)) \times_{J} I \longrightarrow \Pi_{f} (g)$ can be found in $\overline{\mathcal{D}}$. By \ref{horizontalpi}, $\Pi_{\pi_{1}}(e)$ exists and by (the proof of) Proposition 5.2 of \cite{denbesten20}, $\Pi_{\pi_{1}}(e)$ is a weak (strong) homotopy $\Pi$-type for $gd$ and $f$.
  \end{proof}
  
    \begin{lemma}\label{onemorelemma} 
  Write $\mathcal{F}$ and $\mathcal{I}$ for the classes of fibrations and isomorphisms in a path category $\mathcal{C}$. Let $\mathcal{D}$ be a class of display maps with weak (strong) homotopy $\Pi$-types and suppose that $\mathcal{F}$ is the closure of $\mathcal{D} \cup \mathcal{I}$ under composition. Then all weak (strong) homotopy $\Pi$-types exist in $\mathcal{C}$.
  \end{lemma}
  
  \begin{proof}
  This follows from \ref{verticalpi}, since we may assume that $\mathcal{D} = \mathcal{D} \cup \mathcal{I}$ by \ref{addiso}.
  \end{proof}
  
  \begin{theorem}\label{hompitypesinclasscat}
  The classifying category is a path category with weak homotopy $\Pi$-types, which will be strong as soon as function extensionality holds.
\end{theorem}
\begin{proof}
This follows from \ref{pidisplaysyntacticcat} and \ref{onemorelemma}

\end{proof}

\section{Directions for future research}

We have shown that by eliminating the notion of definitional equality from type theory and by replacing all computation rules by propositional equalities, we obtain a system for which type checking is efficiently decidable and which has a natural homotopy-theoretic semantics. This semantics provides strong evidence that most of what happens in the HoTT book \cite{UFP13} can be formalised in a suitable extension of such an objective type theory (with univalent universes and appropriate higher inductive types, for instance). In fact, we think it is very likely that this can be done, but a detailed verification of such a claim would be an enormous amount of work.

From a theoretical perspective there are two important questions. First of all, one would like to understand this phenomenon theoretically by proving an appropriate coherence theorem showing that traditional type theory is conservative over objective type theory. Even formulating such a statement in a mathematically precise way is a non-trivial task; nevertheless, work in this direction has already been done by Valery Isaev \cite{isaev18} and Rafa\"el Bocquet \cite{bocquet20}. 

Another important question is whether objective type theory enjoys \emph{homotopy canonicity}; that is, can one effectively extract from a derivation of $\vdash t \in \mathbb{N}$ in objective type theory a natural number $n$, a term $p$ and a derivation of $\vdash p \in t =_{\mathbb{N}} {\bf succ}^n (0)$? Kapulkin and Sattler have announced a proof of this result for ordinary homotopy type theory: from what we have seen of the proof, we consider it likely that this could be adapted to objective type theory as well, although questions surrounding the effectivity of the proof might remain. 

Finally, there remains the question of whether objective type theory can be of practical importance, as a proof assistant, for instance. The fact that type checking is efficiently decidable would address some of the difficulties pointed out by Geuvers and Wiedijk. Clearly, by storing all the conversions in proof terms, these will become a lot longer than usual and one would have to think carefully about how to manage this complexity. Nevertheless, the success of proof assistants such as HOL, which do store such explicit conversions, means that it should be possible to make such systems practically useful.

\bibliography{hSetoids}

\appendix

\newpage

\section{Rules for objective type theory}

\begin{table}[h] \caption{Rules for identity types} \label{table1}
\begin{center}
\begin{tabular}{c}
Formation Rule \\ \\
$\Gamma \vdash a \in A \qquad \Gamma \vdash b \in A$ \\
\hline
$\Gamma \vdash a =_A b \in \mbox{Type}$ \\ \\
\end{tabular}
\\
\begin{tabular}{c}
Introduction Rule \\ \\
$\Gamma \vdash a \in A$ \\
\hline
$\Gamma \vdash {\bf refl}(A, a) \in a =_A a $ \\ \\
\end{tabular}
\\
\begin{tabular}{c}
Elimination Rule \\ \\
$\Gamma, x \in A, y \in A, u \in x=_A y \, \vdash P \in \mbox{Type}$ \\
$\Gamma \, \vdash p \in a=_A b$ \\
$\Gamma, x \in A \, \vdash d \in P[x, x, {\bf refl}(A, x)/x,y,u]$ \\
\hline
$\Gamma \, \vdash {\bf idrec}(A,[x,y,u]P, a, b, p,[x]d) \in P[a,b,p/x,y,u]$ \\ \\
\end{tabular}
\\
\begin{tabular}{c}
Computation Rule \\ \\
$\Gamma, x \in A, y \in A, u \in x=_A y \, \vdash P \in \mbox{Type}$ \\
$\Gamma \, \vdash a \in A$ \\
$\Gamma, x \in A \, \vdash d \in P[x, x, {\bf refl}(A, x)/x,y,u]$ \\
\hline
$\Gamma \, \vdash {\bf idconv}(A, [x,y,u]P,a,[x]d)$ \\
$\in  {\bf idrec}(A,[x,y,u]P, a,a,{\bf refl}(A, a),[x]d) =_{P[a,a,{\bf refl}(A, a)/x,y,u]} d[a/x]$
\end{tabular}
\end{center}
\end{table}

\newpage

\begin{table}[h] \caption{Rules for $\Pi$-types} \label{table2}
\begin{center}
\begin{tabular}{c}
Formation Rule \\ \\
$\Gamma \vdash A \in {\rm Type} \qquad \Gamma, x \in A \vdash B \in {\rm Type}$ \\
\hline
$\Gamma \vdash \Pi(A, [x]B) \in {\rm Type}$ \\ \\
\end{tabular}
\\
\begin{tabular}{c}
Introduction Rule \\ \\
$\Gamma, x \in A \vdash t \in B$ \\
\hline
$\Gamma \vdash \lambda(A, [x]B, [x]t) \in \Pi(A,[x]B)$ \\ \\
\end{tabular}
\\
\begin{tabular}{c}
Elimination Rule \\ \\
$\Gamma \, \vdash f \in \Pi(A,[x]B) \qquad \Gamma \, \vdash a \in A$ \\
\hline
$\Gamma \, \vdash {\bf app}(A,[x]B,f,a) \in B[a/x]$
\end{tabular}
\\
\begin{tabular}{c}
Computation Rule \\ \\
$\Gamma, x \in A \, \vdash t \in B \qquad \Gamma \, \vdash a \in A$ \\
\hline
$\Gamma \, \vdash {\bf betaconv}(A,[x]B,a,[x]t) \in  {\bf app}(A,[x]B,\lambda(A,[x]B,[x]t),a) =_{B[a/x]} t[a/x]$
\end{tabular}
\end{center}
\end{table}
  
\newpage

\section{Admissable rules for objective type theory}

\begin{table}[h] \caption{Rules for conextual identity types} \label{table5}
\begin{center}
\begin{tabular}{c}
Formation Rule \\ \\
$\Gamma \vdash a \in A \qquad \Gamma \vdash b \in A$ \\
\hline
$\Gamma \vdash a =_A b \in \mbox{Type}$ \\ \\
\end{tabular}
\\
\begin{tabular}{c}
Introduction Rule \\ \\
$\Gamma \vdash a \in A$ \\
\hline
$\Gamma \vdash {\bf refl}(A, a) \in a =_A a $ \\ \\
\end{tabular}
\\
\begin{tabular}{c}
Elimination Rule \\ \\
$\Gamma, x \in A, y \in A, u \in x=_A y, \Delta \, \vdash P \in \mbox{Type}$ \\
$\Gamma, x \in A, \vec{z} \in \Delta[x,x,{\bf refl}(A, x)/x,y,u] \, \vdash d \in P[x, x, {\bf refl}(A, x)/x,y,u]$ \\
$\Gamma \, \vdash p \in a=_A b$ \\
$\Gamma \, \vdash \vec{q} \in \Delta[a,b,p/x,y,u]$ \\
\hline
$\Gamma \, \vdash {\bf idrec}(A,[x,y,u]\Delta,[x, y, u, \vec{z}]P, a, b, p, \vec{q}, [x, \vec{z}]d) \in P[a,b,p,\vec{q}/x,y,u,\vec{z}]$ \\ \\
\end{tabular}
\\
\begin{tabular}{c}
Computation Rule \\ \\
$\Gamma, x \in A, y \in A, u \in x=_A y \, \vdash P \in \mbox{Type}$ \\
$\Gamma, x \in A, \vec{z} \in \Delta[x,x,{\bf refl}(A, x)/x,y,u] \, \vdash d \in P[x, x, {\bf refl}(A, x)/x,y,u]$ \\
$\Gamma \, \vdash a \in A$ \\
$\Gamma \, \vdash \vec{w} \in \Delta[a, a, {\bf refl}(A, a)/x,y,u]$ \\
\hline
$\Gamma \, \vdash {\bf idconv}(A, [x,y,u]P,a,[x]d)$ \\
$ \in  {\bf idrec}(A,[x,y,u]\Delta, [x,y,u,\vec{z}]P, a, a, {\bf refl}(A, a), \vec{w}, [x, \vec{z}]d) =_{P[a,a,{\bf refl}(A, a),\vec{w}/x,y,u,\vec{z}]} d[a, \vec{w}/x, \vec{z}]$
\end{tabular}
\end{center}
\end{table}

\newpage

\begin{table}[h] \caption{Rules for contextual $\Pi$-types} \label{table4}
\begin{center}
\begin{tabular}{c}
Formation Rule \\ \\
$\vdash \Gamma, \Delta \, {\rm Ctxt} \qquad \Gamma, \vec{x} \in \Delta \vdash B \in {\rm Type}$ \\
\hline
$\Gamma \vdash \Pi(\Delta, [\vec{x}]B) \in {\rm Type}$ \\ \\
\end{tabular}
\\
\begin{tabular}{c}
Introduction Rule \\ \\
$\Gamma, \vec{x} \in \Delta \vdash t \in B$ \\
\hline
$\Gamma \vdash \lambda(\Delta, [\vec{x}]B, [\vec{x}]t) \in \Pi(\Delta,[\vec{x}]B)$ \\ \\
\end{tabular}
\\
\begin{tabular}{c}
Elimination Rule \\ \\
$\Gamma \, \vdash f \in \Pi(\Delta,[\vec{x}]B) \qquad \Gamma \, \vdash \vec{a} \in \Delta$ \\
\hline
$\Gamma \, \vdash {\bf app}(\Delta,[\vec{x}]B,f,\vec{a}) \in B[\vec{a}/\vec{x}]$ \\ \\
\end{tabular}
\\
\begin{tabular}{c}
Computation Rule \\ \\
$\Gamma, \vec{x} \in \Delta \, \vdash t \in B \qquad \Gamma \, \vdash \vec{a} \in \Delta$ \\
\hline
$\Gamma \, \vdash {\bf betaconv}(\Delta,[\vec{x}]B,a,[\vec{x}]t) \in  {\bf app}(\Delta,[\vec{x}]B,\lambda(\Delta, [\vec{x}]B, [\vec{x}]t),\vec{a}) =_{B[\vec{a}/\vec{x}]} t[\vec{a}/\vec{x}]$
\end{tabular}
\end{center}
\end{table}

\end{document}